\documentclass{article}
\usepackage{arxiv}

\usepackage{amsthm}

\usepackage{url}
\usepackage{fix-cm}
\usepackage{xypic}  
\usepackage{wrapfig}

\usepackage{xcolor}  

\usepackage{amssymb}

\usepackage{amsmath}
\usepackage{mathdots}

\usepackage{mathtools}

\usepackage{tensor}
\usepackage{xifthen}

\usepackage{urwchancal}
\DeclareFontFamily{OT1}{pzc}{}
\DeclareFontShape{OT1}{pzc}{m}{it}{<-> s * [1.10] pzcmi7t}{}
\DeclareMathAlphabet{\mathpzc}{OT1}{pzc}{m}{it}

\usepackage{graphicx}
\usepackage{ifpdf}
\ifpdf
\usepackage{epstopdf}
\PrependGraphicsExtensions{.svg}
\fi

\newtheorem{theorem}{Theorem}[section]
\newtheorem{lemma}[theorem]{Lemma}

\newtheorem{proposition}[theorem]{Proposition}

\newtheorem{definition}[theorem]{Definition}

\providecommand{\R}{\mathbb{R}}

\providecommand{\SO}{\mathbf{SO}}

\providecommand{\SE}{\mathbf{SE}}

\providecommand{\MR}{\mathbf{MR}}

\providecommand{\VSLAM}{\mathbf{VSLAM}}

\providecommand{\grpG}{\mathbf{G}}

\providecommand{\gothmr}{\mathfrak{mr}}

\providecommand{\so}{\mathfrak{so}}
\providecommand{\se}{\mathfrak{se}}

\providecommand{\vslam}{\mathfrak{vslam}}

\providecommand{\Sph}{\mathrm{S}}

\providecommand{\calM}{\mathcal{M}}
\providecommand{\calN}{\mathcal{N}}

\providecommand{\calT}{\mathcal{T}}

\providecommand{\vecV}{\mathbb{V}}

\providecommand{\tT}{\mathrm{T}} 

\DeclareMathOperator{\Ad}{Ad}

\DeclareMathOperator*{\argmin}{argmin}

\providecommand{\id}{\mathrm{id}} 

\providecommand{\td}{\mathrm{d}}
\providecommand{\tD}{\mathrm{D}}

\providecommand{\ddt}{\frac{\td}{\td t}}

\def\frameZero{\mbox{$\{0\}$}}

\providecommand{\mr}[1]{\mathring{#1}} 
\providecommand{\ub}[1]{\underline{#1}}

\providecommand{\ob}[1]{\overline{#1}} 
\usepackage{accents}
\usepackage{mathtools}
\makeatletter
\providecommand{\scirc}{%
    \hbox{\fontfamily{\rmdefault}\fontsize{0.4\dimexpr(\f@size pt)}{0}\selectfont{\raisebox{-0.52ex}[0ex][-0.52ex]{$\circ$}}}}

\makeatother

\makeatletter
\providecommand{\ucirc}{%
    \hbox{\fontfamily{\rmdefault}\fontsize{0.4\dimexpr(\f@size pt)}{0}\selectfont{\raisebox{0.0ex}[0ex][-0.52ex]{$\circ$}}}}

\makeatother

\mathchardef\mhyphen="2D

\providecommand{\idx}[5][]{
\ifthenelse{\isempty{#1}}
{\tensor*[_{#4}^{#3}]{#2}{_{#5}}}
{\tensor*[_{#4}^{#3}]{#2}{^{#1}_{#5}}}
}

\providecommand{\etal}{\textit{et al.~}}

\renewcommand{\mr}[1]{#1^\circ}

\usepackage{natbib}

\usepackage{hyperref}

\newcommand{\papercitation}{
van Goor, Pieter, et al. "Constructive observer design for visual simultaneous localisation and mapping." \textit{Automatica} 132 (2021): 109803.
}

\newcommand{\publicationdetails}
{
\papercitation \\
\copyrightNoticeIFACAccepted{2021} 
\\
This research was supported by the Australian Research Council through the “Australian Centre of Excellence for Robotic Vision” CE140100016.
}
\newcommand{\publicationversion}
{Author accepted version}

\begin{document}



\headertitle{Constructive Observer Design for Visual Simultaneous Localisation and Mapping}
\title{Constructive Observer Design for Visual Simultaneous Localisation and Mapping}



\author{
    \href{https://orcid.org/0000-0003-4391-7014}{\includegraphics[scale=0.06]{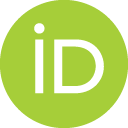}\hspace{1mm}
    Pieter van Goor}
\\
    Australian Centre for Robotic Vision \\
	Australian National University \\
    ACT, 2601, Australia \\
    \texttt{Pieter.vanGoor@anu.edu.au} \\
	\And	\href{https://orcid.org/0000-0002-7779-1264}{\includegraphics[scale=0.06]{orcid.png}\hspace{1mm}
    Tarek Hamel}
\\
    I3S (University C\^ote d'Azur, CNRS, Sophia Antipolis) \\
    and Insitut Universitaire de France \\
    \texttt{THamel@i3s.unice.fr} \\
	\And	\href{https://orcid.org/0000-0002-7803-2868}{\includegraphics[scale=0.06]{orcid.png}\hspace{1mm}
    Robert Mahony}
\\
    Australian Centre for Robotic Vision \\
	Australian National University \\
    ACT, 2601, Australia \\
	\texttt{Robert.Mahony@anu.edu.au} \\
    \And	\href{https://orcid.org/0000-0002-5881-1063}{\includegraphics[scale=0.06]{orcid.png}\hspace{1mm}
    Jochen Trumpf}
\\
    Australian Centre for Robotic Vision \\
	Australian National University \\
    ACT, 2601, Australia \\
	\texttt{Jochen.Trumpf@anu.edu.au} \\
}

\maketitle

\vspace{1cm}

\begin{abstract}
Visual Simultaneous Localisation and Mapping (VSLAM) is a well-known problem in robotics with a large range of applications.
This paper presents a novel approach to VSLAM by lifting the observer design to a novel Lie group $\VSLAM_n(3)$ on which the system output is equivariant.
The perspective gained from this analysis facilitates the design of a non-linear observer with almost semi-globally asymptotically stable error dynamics.
Simulations are provided to illustrate the behaviour of the proposed observer and experiments on data gathered using a fixed-wing UAV flying outdoors demonstrate its performance.
\end{abstract}


\section{Introduction}\label{sec:intro}

\begin{figure}[!htb]
    \centering
    \includegraphics[width=0.75\linewidth]{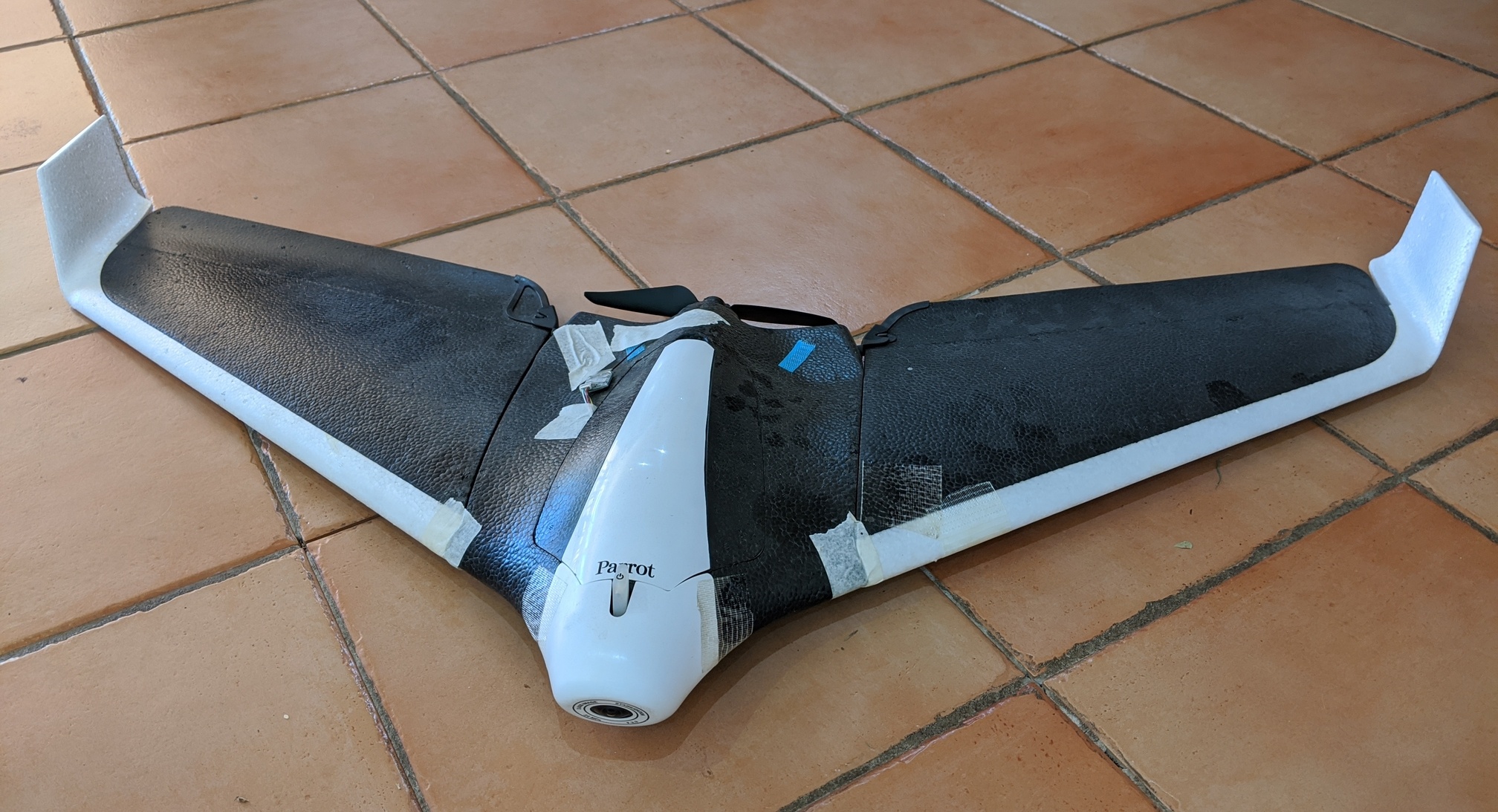}
    \caption{The Disco Parrot UAV used to gather video data to test the proposed observer.
}
    \label{fig:disco_parrot_uav}
\end{figure}

Simultaneous Localisation and Mapping (SLAM) has been an established problem in mobile robotics for at least the last 30 years \cite{2015_Manuel_vslam}.
Visual SLAM (VSLAM) refers to the special case where the only exteroceptive sensors available are cameras, and is frequently used to refer to the challenging situation where only a single monocular camera is available.
The inherent non-linearity of the VSLAM problem remains challenging \cite{2016_Cadena_TRO} and state-of-the-art solutions suffer from high computational complexity and poor scalability \cite{2015_Manuel_vslam,2012_Strasdat_CVIU}.
Due to the low cost and low weight, as well as the ubiquity of single camera systems, the VSLAM problem remains an active research topic \cite{2015_Manuel_vslam,2018_Delmerico_icra}.

Both the SLAM and VSLAM problems have recently attracted interest in the non-linear observer community, drawing from earlier work on attitude estimation \cite{2008_Mahony_tac,2008_Bonnabel_TAC} and pose estimation \cite{2009_Baldwin_icra,2010_Vasconcelos.SCL,RM_2011_Hua_cdc}.
Barrau and Bonnabel \cite{2016_Barrau_arxive} exploited a novel Lie group to design an invariant Extended Kalman Filter for the SLAM problem.
Parallel work by Mahony and Hamel \cite{2017_Mahony_cdc} proposed the same group structure along with a novel quotient manifold structure for the state-space of the SLAM problem.
Work by Zlotnik and Forbes \cite{2018_forbes_TAC} derives a geometrically motivated observer for the SLAM problem that includes estimation of bias in linear and angular velocity inputs.
For the VSLAM problem, where only bearing measurements are available, Lourenco \etal \cite{2016_LouGueBatOliSil,2018_Lourenco_RAS} proposed an observer with a globally exponentially stable error system using depths of landmarks as separate components of the observer.
Bjorne \etal \cite{2017_bjorne_fusion} use an attitude heading reference system (AHRS) to determine the orientation of the robot, and then solve the SLAM problem using a linear Kalman filter.
A similar approach to VSLAM is undertaken by Le Bras \etal \cite{2017_LebHamMahSam}.
Hamel and Samson \cite{2018_Hamel_TAC} have also introduced a Riccati observer  for the case where the orientation of the robot is known.
Recent work by the authors \cite{2019_vangoor_cdc_vslam} introduced a new symmetry structure specifically targeting the  VSLAM problem but used an observer design that was a lifted version of that proposed in \cite{2018_Hamel_TAC}.

In this paper we present a novel non-linear equivariant observer for the VSLAM problem.
The approach uses the SLAM manifold state-space proposed in \cite{2017_Mahony_cdc} along with a novel symmetry Lie-group, $\VSLAM_n(3)$, introduced by van Goor \etal \cite{2019_vangoor_cdc_vslam} but fully developed for the first time in this paper.
We extend the results of \cite{2019_vangoor_cdc_vslam} by providing equivariant group actions on the state and output spaces leading to the definition of the lifted system, a lifted observer and more importantly  an intrinsic error that is globally defined.
We propose a
Lyapunov function expressed in the intrinsic error coordinates and use this to construct an observer for the visual SLAM problem posed on the symmetry group $\VSLAM_n(3)$.
This is in contrast to the majority of state-of-the-art algorithms which depend on local error coordinates and local linearisation \cite{2016_Cadena_TRO}.
The recent IEKF results \cite{2018_brossard_iros} exploit a global symmetry of the state-space, however, the symmetry used is not compatible with visual bearings and the resulting algorithm still depends on local linearisation of the measurement function.
In our proposed algorithm, separate constant gains for landmark bearing and depth estimates are used, making the design algebraically simple and leading to low computational cost.
We show that the error dynamics are \emph{almost semi-globally asymptotically stabilisable} (Def.~\ref{def:semiGAS}).
The resulting algorithm has low computation and memory requirements, making it ideally suited to embedded systems applications in consumer electronics.

This paper consists of six sections alongside the introduction and conclusion.
Section \ref{sec:preliminaries} introduces key notation and identities.
In Section \ref{sec:problem-formulation}, we formulate the kinematics, state-space and output of the VSLAM system, and in Section \ref{sec:symmetry} we introduce the Lie group $\VSLAM_n(3)$ and its actions on the state and output spaces.
In Section \ref{sec:observer} we derive a non-linear observer on the Lie group, and in Sections \ref{sec:simulations} and \ref{sec:experiments} we provide the results of a simulation and a real-world experiment carried out using a Disco Parrot UAV (Figure \ref{fig:disco_parrot_uav}).
The principal contribution of the paper is theoretical and the experimental sections support this by illustrating the properties of the algorithm and demonstrating that it functions on real-world data.
We do not aspire to provide a comprehensive benchmark of performance against state-of-the-art SLAM systems in the present paper.

\section{Preliminaries} \label{sec:preliminaries}

The special orthogonal and special Euclidean matrix Lie groups are denoted $\SO(3)$ and $\SE(3)$, respectively, with Lie algebras $\so(3)$ and $\se(3)$.
For any column vector $\Omega = (\Omega_1, \Omega_2, \Omega_3) \in \R^3$, the corresponding skew-symmetric matrix is denoted
\begin{align*}
\Omega^\times := \left( \begin{matrix}
0 & -\Omega_3 & \Omega_2 \\
\Omega_3 & 0 & -\Omega_1 \\
-\Omega_2 & \Omega_1 & 0
\end{matrix} \right) \in \so(3).
\end{align*}
This matrix has the property that, for any $v \in \R^3$, $\Omega^\times v = \Omega \times v$ where $\Omega \times v$ is the vector (cross) product between $\Omega$ and $v$.
For any unit vector $y \in \Sph^2 \subset \R^3$ and any vector $v \in \R^3$,
\begin{align}
y^\times y^\times v &= yy^\top v - v. \label{eq:cross2_projector_formula}
\end{align}

Consider a homogeneous matrix $P \in \SE(3)$.
The notation $R_P \in \SO(3)$ and $x_P \in \R^3$ is used to represent the rotation and translation components of $P$, respectively; that is
\begin{align*}
P = \begin{pmatrix}
R_P & x_P \\ 0 & 1
\end{pmatrix} \in \SE(3).
\end{align*}
Likewise, for a matrix $U \in \se(3)$, the notation $\Omega_U^\times \in \so(3)$, with $\Omega_U \in \R^3$,
and $V_U \in \R^3$ represent the angular and linear velocity components of $U$, respectively; i.e.
\begin{align*}
U = \begin{pmatrix}
\Omega_U^\times & V_U \\ 0 & 0
\end{pmatrix} \in \se(3).
\end{align*}

For a background on smooth manifolds, Lie groups and their actions, the authors recommend \cite[Chapter 7]{2013_lee_manifolds}.

\section{Problem Formulation} \label{sec:problem-formulation}

\subsection{VSLAM State Space}

Fix an arbitrary reference frame $\{0\}$.
Let $P \in \SE(3)$ and $p_i \in \R^3, i=1,...,n$ represent the robot pose and landmark coordinates, respectively, defined with respect to $\{0\}$.
The raw coordinates of the SLAM problem are written $(P, p_1,..., p_n) \in \SE(3) \times \R^3 \times \cdots \times \R^3$.
The notation $(P, p_i) \equiv (P, p_1,..., p_n)$ is used for simplicity in the sequel.


The physical measurements in a monocular VSLAM system are the bearings (3D directions) of landmarks perceived by the robot.
We assume from now on that the observed landmarks and the robot are not collocated to ensure that the bearing measurements are well defined.
Interestingly, this assumption has a substantive impact on the nature of the global symmetries that can be admitted.
We make this assumption explicit, defining the \emph{total space} $\mr{\calT}_n(3)$ of SLAM configurations considered to be
\begin{align}
\mr{\calT}_n(3) = \big\{ (P,p_i) \in \SE(3) \times \R^3 \times \cdots \times \R^3 \ \vline 
 p_i \neq x_P, i=1,...,n \big\}.
\label{eq:reduced_total_space}
\end{align}
We assume that the trajectory of the robot remains in $\mr{\calT}_n(3)$ for all time.

Given two configurations $(P, p_i), (P', p'_i) \in \mr{\calT}_n(3)$, then $(P, p_i) \simeq (P', p'_i)$ if there exists $S \in \SE(3)$ such that $(P, p_i) = (S^{-1} P', R_S^\top(p'_i - x_S))$.
That is, two sets of coordinates in the total space $\mr{\calT}_n(3)$ are considered equivalent when they are related by a rigid body transformation of the reference frame $\frameZero$.
It is straightforward to show the relation $\simeq$ is an equivalence relation on $\mr{\calT}_n(3)$ and
\begin{align*}
\lfloor P, p_i \rfloor  := \left\{ (S^{-1}P, R_S^\top(p_i - x_S)) \ \vline \ S \in \SE(3) \right\}.
\end{align*}
is the associated equivalence class.
The VSLAM manifold is the set
\begin{align*}
\mr{\calM}_n(3) = \left\{ \lfloor P, p_i \rfloor \ | \ (P, p_i) \in \mr{\calT}_n(3) \right\},
\end{align*}
with quotient manifold structure.
This is the open subset of the SLAM manifold $\calM_n(3)$ considered in \cite{2017_Mahony_cdc} without those equivalence classes where a landmark is co-located with the robot.
Two configurations are equivalent on the SLAM manifold, $(P^1, p_i^1) \simeq (P^2, p_i^2)$, if and only if the ego-centric coordinates of the landmarks are equal, $R_{P^1}^\top(p^1_i - x_{P^1}) = R_{P^2}^\top(p^2_i - x_{P^2})$ for all $i$.

\subsection{VSLAM Kinematics}
Assume that the robot is moving in a static environment.
Define a velocity input vector space $\vecV = \se(3)$ to contain the rigid-body velocity $U \in \se(3)$ of the robot.
The kinematics of the VSLAM system are given by
the system function $f: \mr{\calT}_n(3) \times \vecV \to \tT \mr{\calT}_n(3)$,
\begin{align}
    \frac{\td}{\td t}(P, p_i) &= f((P, p_i), U), \notag \\
    &:= (PU, 0). \label{eq:input_function_f}
\end{align}

\subsection{System Output}

The camera measurements are modelled as elements of the sphere $\Sph^2$.
Each individual bearing is given by an output function $ h^i : \mr{\calT}_n(3) \to \Sph^2$,
\begin{align}
h^i(P, p_i) := \frac{R_P^\top(p_i - x_P)}{\Vert p_i - x_P \Vert}, \label{eq:output_function_h_i}
\end{align}
The output functions are well defined on $\mr{\calM}_n(3)$ since
\begin{align}
h^i(S^{-1} P, R_S^\top(p_i - x_S))
&=\frac
{(R_S^\top R_P)^\top (R_S^\top(p_i - x_S) - R_S^\top(x_P - x_S))}
{\Vert R_S^\top(p_i - x_S) - R_S^\top(x_P - x_S) \Vert}, \notag \\
&= \frac{R_P^\top(p_i - x_P)}{\Vert p_i - x_P \Vert}.  \label{eq:welldef-of-body-coordinates}
\end{align}
The full output space of the VSLAM system is defined as a product of $n$ spheres,
\begin{align*}
\calN^n(3) := \Sph^2 \times \cdots \times \Sph^2,
\end{align*}
with a combined output function $h : \mr{\calT_n(3)} \to \calN^n(3)$
\begin{align}
h ( P, p_i ) := \left( \frac{R_P^\top(p_1 - x_P)}{\Vert p_1 - x_P \Vert}, \ldots,  \frac{R_P^\top(p_n - x_P)}{\Vert p_n - x_P \Vert} \right).
\label{eq:output_function_h}
\end{align}

\section{Symmetry of the VSLAM Problem} \label{sec:symmetry}
\subsection{Symmetry of the Total Space \texorpdfstring{$\mr{\calT}_n(3)$}{T0n(3)}}

Define a Lie group
\begin{align*}
\VSLAM_n(3) = \SE(3) \times (\SO(3) \times \MR(1))^n,
\end{align*}
with product Lie group structure, where $\MR(1)$ is the multiplicative real group of positive real numbers.
This Lie group was first proposed in van Goor \etal \cite{2019_vangoor_cdc_vslam}.
The associated Lie algebra is denoted $\vslam_n(3)$.
We write elements of $\VSLAM_n(3)$ as $(A, (Q, a)_i) \equiv (A, (Q, a)_1, ..., (Q,a)_n) \in \VSLAM_n(3)$.
The group product, identity and inverse are given by
\begin{gather*}
    (A_1, (Q_1,a_1)_i) \cdot (A_2, (Q_2,a_2)_i) = (A_1 A_2, (Q_1 Q_2, a_1 a_2)_i), \\
    \id = (I_4, (I_3, 1)_i), \qquad (A, (Q,a)_i)^{-1} = (A^{-1}, (Q^\top,a^{-1})_i).
\end{gather*}

\begin{lemma}
The mapping $\Upsilon : \VSLAM_n(3) \times \mr{\calT}_n(3) \to \mr{\calT}_n(3)$ defined by
\begin{align}\label{eq:group_action_upsilon}
\Upsilon((A,(Q,a)_i),(P,p_i))
:= (PA, a_i^{-1}R_{PA}Q_i^\top R_P^\top(p_i-x_P) +x_{PA}),
\end{align}
is a transitive right group action of $\VSLAM_n(3)$ on $\mr{\calT}_n(3)$.
\end{lemma}
\begin{proof}
Trivially, $\Upsilon((I_4, (I_3, 1)_i), (P, p_i)) = (P, p_i)$ for any $(P, p_i) \in \mr{\calT}_n(3)$.
Let $(A_1,(Q_1,a_1)_i), (A_2, (Q_2,a_2)_i) \in \VSLAM_n(3)$ and $(P, p_i) \in \mr{\calT}_n(3)$ be arbitrary.
Then
\begin{align*}
\Upsilon((A_1, (Q_1,a_1)_i), \Upsilon((A_2,(Q_2,a_2)_i), (P,p_i)))
&= (PA_2A_1, (a_1^{-1}a_2^{-1} R_{PA_2A_1} Q_1^\top  \\
&\hspace{1cm} Q_2^\top R_P^\top(p-x_P))+ x_{PA_2A_1})_i), \\
&= \Upsilon((A_2A_1, (Q_2 Q_1, a_2a_1)_i), (P, p_i)), \\
&= \Upsilon((A_2,(Q_2,a_2)_i) \cdot (A_1,(Q_1,a_1)_i), (P, p_i)).
\end{align*}
Thus $\Upsilon$ is a group action.
To see that $\Upsilon$ is transitive, let $(P, p_i), (P', p_i') \in \mr{\calT}_n(3)$ be arbitrary.
Choose $(A, (Q,a)_i) \in \VSLAM_n(3)$ to satisfy
\begin{gather*}
    A = P^{-1} P', \qquad a_i = \frac{\Vert p_i - x_{P} \Vert}{\Vert p'_i - x_{P'} \Vert}, \\
    Q_i \frac{R_{P'}^\top (p_i' - x_{P'})}{\Vert p_i' - x_{P'} \Vert} = \frac{R_{P}^\top (p_i - x_{P})}{\Vert p_i - x_{P} \Vert}.
\end{gather*}
Then $\Upsilon((A,(Q,a)_i), (P,p_i)) = (P',p_i')$.
\end{proof}

The action $\Upsilon$ of $\VSLAM_n(3)$ on $\mr{\calT}_n(3)$ is shown in Figure \ref{fig:symmetry_action}.
Given $(A,(Q_A,a)_i) \in \VSLAM_n(3)$ and $(P,p_i) \in \mr{\calT}_n(3)$, the action transforms the robot pose $P$ by right-translation by $A$.
The landmark points are transformed by considering their body-fixed coordinates; applying a rotation $Q_i^\top$ and scaling $a_i^{-1}$; transforming them along with the robot pose; and finally writing the result in inertial coordinates.

\begin{figure}[ht]
    \begin{centering}
    \includegraphics[width = 0.8\linewidth]{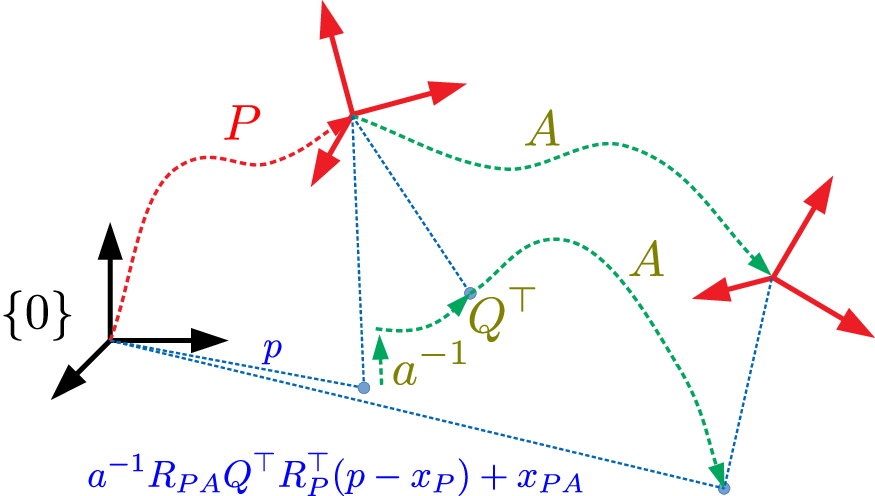}
    \caption{The action $\Upsilon$ of the VSLAM group on the state space for some given $(A,(Q_A,a)_i) \in \VSLAM_n(3)$ and $(P,p_i) \in \mr{\calT}_n(3)$.
    The pose $P$ is mapped to $P A$.
    The body fixed frame landmark points $R_P^\top (p_i - x_P)$ are rotated by $Q^\top$ and scaled by $a^{-1}$ in the body-fixed frame before \emph{transforming with the robot pose} to a new point $p_i'$ which is then rewritten in the inertial frame.
    }
    \label{fig:symmetry_action}
    \end{centering}
\end{figure}

\subsection{Symmetry of the Output Space}
There is an action $\rho$ of the $\VSLAM_n(3)$ group on the output space such that the measurement function $h$ defined in \eqref{eq:output_function_h} is equivariant with respect to $\Upsilon$ and $\rho$.
The following lemmas define this action $\rho$ and show the equivariance of $h$ for the proposed $\VSLAM_n(3)$ geometry.
The authors know of no output action that makes bearing outputs equivariant with respect to prior geometries proposed for SLAM \cite{2016_Barrau_arxive,2017_Mahony_cdc}.
The equivariance structure enables the development of a globally-defined intrinsic error.

\begin{proposition}
The mapping $\rho : \VSLAM_n(3) \times \calN^n(3) \to \calN^n(3)$ defined by
\begin{align} \label{eq:group_action_rho}
\rho((A,&(Q,a)_i),y_i) = Q_i^\top y_i,
\end{align}
is a right group action of $\VSLAM_n(3)$ on $\calN^n(3)$.
\end{proposition}
\begin{proof}
The proof is straightforward.
\end{proof}

\begin{lemma}
The output $h:\mr{\calT}_n(3) \to \calN^n(3)$ \eqref{eq:output_function_h} is equivariant with respect to actions $\Upsilon$ \eqref{eq:group_action_upsilon} and $\rho$ \eqref{eq:group_action_rho}.
That is, for any $X \in \VSLAM_n(3)$ and any $\xi \in \mr{\calT}_n(3)$,
\begin{align*}
    h(\Upsilon(X, \xi)) = \rho(X, h(\xi)).
\end{align*}
\end{lemma}

\begin{proof}
Let $X = (A, (Q,a)_i)$ and $\xi = (P, p_i)$ be arbitrary.
Then
\begin{align*}
h(\Upsilon(X, \xi))
&= h\left( PA, a_i^{-1}R_{PA}Q_i^\top R_P^\top(p_i-x_P) +x_{PA}  \right), \\
&=  \frac{R_{PA}^\top ( a_i^{-1}R_{PA}Q_i^\top R_P^\top(p_i-x_P) +x_{PA} - x_{PA})}{\Vert a_i^{-1}R_{PA}Q_i^\top R_P^\top(p_i-x_P) +x_{PA}  - x_{PA} \Vert}, \\
&= \frac{Q_i^\top R_P^\top(p_i-x_P)}{\Vert p_i-x_P \Vert} , \\
&= \rho(X, h(\xi)).
\end{align*}
\end{proof}

\subsection{Lift of the VSLAM Kinematics}
In order to consider the system on the $\VSLAM_n(3)$ group, the kinematics from the total space must be lifted onto the group.
A lift is a map $\Lambda: \mr{\calT}_n(3) \times \vecV \to \vslam_n(3)$ such that
\begin{align} \label{eq:lift_condition}
\tD \Upsilon_{(P, p_i)}(\id) \left[ \Lambda((P, p_i), U) \right] = f((P, p_i), U)
\end{align}
where $ f((P, p_i), U)$ is given by \eqref{eq:input_function_f}.

\begin{lemma} \label{lem:velocity_lift}
The function $\Lambda: \mr{\calT}_n(3) \times \vecV \to \vslam_n(3)$, defined by
\begin{align}
\Lambda((P, p_i), U) := (U, (\Lambda_Q (U, R_P^\top(p_i - x_P)), \;
\Lambda_a (U, R_P^\top(p_i - x_P)) )), \label{eq:velocity_lift}
\end{align}
where $\Lambda_Q : \se(3) \times (\R^3 \setminus \{0\}) \to \so(3)$ is given by
\[
\Lambda_Q \left( U, q \right) := \left( \Omega_U + \frac{q \times V_U}{|q|^2}\right)^\times,
\]
and $\Lambda_a : \se(3) \times (\R^3 \setminus \{0\}) \to \gothmr(1)$ is given by
\begin{align*}
\Lambda_a \left( U, q \right) := \frac{q^\top V_U}{|q|^2},
\end{align*}
is a lift in the sense of \eqref{eq:lift_condition} of the kinematics \eqref{eq:input_function_f} onto $\vslam_n(3)$ with respect to the group action \eqref{eq:group_action_upsilon}.
\end{lemma}

\begin{proof}
Given $(P,p_i) \in \mr{\calT}_n(3)$, let $q_i := R_P^\top(p_i - x_P)$, $W^\times_i := \Lambda_Q(U, q_i)$, and $w_i := \Lambda_a(U, q_i)$.
Then
\begin{align}
\tD \Upsilon_{(P, p_i)}(\id) \left[\Lambda((P, p_i), U) \right]
& = \tD \Upsilon_{(P, p_i)}(\id) \left[ (U, (W^\times_i, w_i)) \right] \notag \\
& = \left( PU, v_i \right) \label{eq:equi_formula_proof}
\end{align}
where $v_i$ is
\begin{align*}
v_i = \left. \frac{\td}{\td s} \right|_{s=0}
\left[ (1+s w_i)^{-1} R_{P(I_4+sU)} (I_3+s W_i^\times)^\top q_i
 + x_{P(I_4+sU)}   \right].
\end{align*}
Computing this derivative and evaluating at $s = 0$ one obtains
\begin{align}
v_i & = -w_i R_{P} q_i + R_{P}\Omega_U^\times q_i + R_{P} (W_i^\times)^\top q_i + R_{P} V_U, \notag  \\
&= -\frac{q_i^\top V_U}{|q_i|^2} R_{P} q_i + R_{P}\Omega_U^\times q_i + R_{P} V_U 
 - R_{P} \left( \Omega_U + \frac{q_i \times V_U}{|q_i|^2}\right)^\times q_i, \label{eq:sub_lambda} \\
&=  -R_{P} \frac{q_i q_i^\top}{|q_i|^2} V_U + R_{P} V_U + R_{P} \frac{q_i q_i^\top}{|q_i|^2} V_U - R_P V_U, \label{eq:cancel_and sub}\\
&= 0, \notag
\end{align}
where \eqref{eq:sub_lambda} follows from substituting for $w_i$ and $W_i$ with the full expressions for $\Lambda_a$ and $\Lambda_Q$, and \eqref{eq:cancel_and sub} follows from cancelling the $\Omega_U$ term and using the relationship \eqref{eq:cross2_projector_formula} as well as rearranging the first term.
The result follows from substituting directly into \eqref{eq:equi_formula_proof} and recalling \eqref{eq:input_function_f}.
\end{proof}

\section{Observer Design} \label{sec:observer}

Figure \ref{fig:schematic} shows a schematic overview of the proposed observer.
The key features of equivariant observer design are the distinction between the observer state $\hat{X} \in \VSLAM_n(3)$ and the estimated state $(\hat{P}, \hat{p}_i) \in \calT_n^\circ(3)$, and the design of the correction term around the output error $(d_i) = \rho(\hat{X}^{-1}, y_i)$ rather than the raw measurements $(y_i) \in \calN^n(3)$.
\begin{figure*}[!htb]
    \centering
    \includegraphics[width=0.8\linewidth]{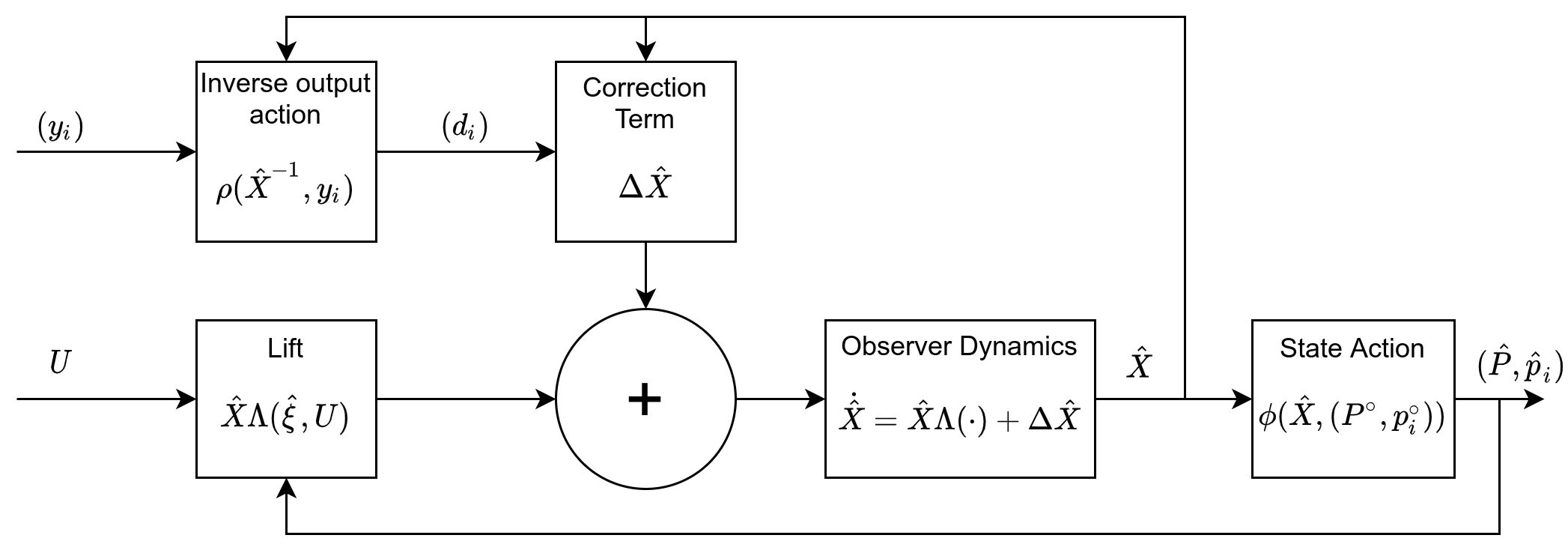}
    \caption{An overview of the observer system.
    Note the distinction between the observer state $\hat{X}$ and the estimated state $(\hat{P}, \hat{p}_i)$, and that the correction term $\Delta$ is based on the output error $(d_i)$ rather than the raw measurement $(y_i)$.}
    \label{fig:schematic}
\end{figure*}

\subsection{Lifted System Kinematics}
Let $\xi = (P, p_i) \in \mr{\calT}_n(3)$ denote the true state of the VSLAM system.
The kinematics of $\xi$ are given by \eqref{eq:input_function_f}.
Choose an arbitrary origin configuration $\mr \xi = (\mr P, \mr{p}_i) \in \mr{\calT}_n(3)$.
The lifted system is (Lemma~\ref{lem:velocity_lift})
\begin{align}
\frac{\td}{\td t} X &= X \Lambda(\Upsilon(X, \mr{\xi}), U),
\label{eq:lifted_system_state} \\
&= (A U, (Q \Lambda_Q, a \Lambda_a)_i), \notag
\end{align}
for $X = (A, (Q,a)_i) \in \VSLAM_n(3)$.
If $\Upsilon(X(0), \mr{\xi}) = \xi(0)$ then trajectories of the lifted system kinematics project to trajectories of the VSLAM kinematics \eqref{eq:input_function_f} \cite{RM_2013_Mahony_nolcos}.
That is, recalling \eqref{eq:group_action_upsilon}
\begin{align*}
(P(t),p_i(t))  = \Upsilon(X(t), \mr{\xi})
 = (\mr{P} A, a_i^{-1} R_{\mr{P} A} Q_i^\top R_{\mr{P}}^\top(\mr{p}_i-x_{\mr{P}}) + x_{\mr{P}A})
\end{align*}
for all $t \geq 0 $, where we have dropped the time dependence of $A$, $Q_i$ and $a_i$ from the notation to improve readability.

\subsection{Landmark Observer}

Let $X = (A, (Q,a)_i)$ be a trajectory of the lifted system \eqref{eq:lifted_system_state} associated with a trajectory  $\xi = (P,p_i)$ of the true system satisfying \eqref{eq:input_function_f} for measured input signal $U = (\Omega_U^\times, V_U)$.
Let $y_i := h^i(\xi)$ \eqref{eq:output_function_h_i} denote the output.

Fix an arbitrary origin configuration $\mr{\xi} = (\mr{P}, \mr{p}_i) \in \mr{\calT}_n(3)$.
The observer state is $\hat{X} = (\hat{A}, (\hat{Q}, \hat{a})_i) \in \VSLAM_n(3)$, with kinematics given by
\begin{align} \label{eq:group_observer_state}
\frac{\td}{\td t} \hat{X} &:= \hat{X} \Lambda(\Upsilon(\hat{X}, \mr \xi), U) - \Delta_{\hat{X}} \hat{X}, \notag \\
\hat{X}(0) &= \id,
\end{align}
where $\Lambda$ is the lift defined by \eqref{eq:velocity_lift} and $\Delta_{\hat{X}} = (\Delta, (\Gamma, \gamma)_i) \in \vslam_n(3)$ is a  correction term that is chosen later.
The state estimate is given by $\hat{\xi} = (\hat{P}, \hat{p}_i) = \Upsilon(\hat{X}, \mr{\xi})$  \eqref{eq:group_action_upsilon}.
Let $\mr{y}_i := h^i(\mr{\xi})$ denote the origin output, and let $d_i$ denote the output error \cite{RM_2013_Mahony_nolcos}, defined as
\begin{align} \label{eq:output_error}
d_i &= \rho(\hat{X}^{-1}, y_i) = \rho(X\hat{X}^{-1}, \mr{y}_i).
\end{align}
Define the true range $r_i = \Vert {p}_i - x_{{P}} \Vert$ and estimated range $\hat{r}_i = \Vert \hat{p}_i - x_{\hat{P}} \Vert$ for each $i$.
The range error is
\begin{align} \label{eq:range_error}
    \tilde{r}_i = \frac{\hat{r}_i}{r_i},
\end{align}
for each $i$.

Define the twice differentiable barrier function $\beta^\epsilon_{\ub{c}}: (\epsilon, \infty) \to [0, \infty)$ to be
\begin{align} \label{eq:beta_function}
    \beta^\epsilon_{\ub{c}} (c) &:= \begin{cases}
        \frac{(c - \ub{c})^2}{(\ub{c} - \epsilon)^2(c - \epsilon)}, & \epsilon < c < \ub{c} \\
        0, & c \geq \ub{c}
        \end{cases},
\end{align}
for parameters $0 < \epsilon < \ub{c}$.

The landmark correction terms $\Gamma_i$ and $\gamma_i$ for $\Delta_{\hat{X}}$ are defined to be
\begin{align} \label{eq:landmark_corrections}
\Gamma_i &:= \left(\frac{d_i^\top\hat{Q}_i V_U}{\hat{r}_i(1 + d_i^\top \mr{y}_i)}
-\frac{k_i}{{(1+ d_i^\top \mr{y}_i)^2}} \right) (d_i^\times \mr{y}_i)^\times \notag \\
&\qquad +\frac{1}{\hat{r}}\left( (\mr{y}_i-d_i)^\times \hat{Q}_iV_U\right)^\times, \notag \\
\gamma_i &:= \frac{\alpha_i}{\hat{r}_i^2}\left((1-d_i^\top \mr{y}_i) d_i^\top \hat{Q}_i V_U - {\mr{y}}_i^{\top} (d_i \times \hat{Q}_i V_U)^\times d_i \right) \notag \\
&\qquad + \frac{1}{\hat{r}_i}(\mr{y}_i - d_i)^\top \hat{Q}_i V_U
+ \frac{\alpha_i}{\hat{r}_i} \beta^\epsilon_{\ub{c}} (\hat{r}_i),
\end{align}
where $k_i$ and $\alpha_i$ are constant positive scalars.
The robot pose correction term $\Delta \in \SE(3)$ can be chosen to be any continuous function of the observer state and measurements (\textit{cf}.~\S\ref{sec:Delta}).

\begin{theorem} \label{thm:landmark_observer}
Consider the observer $\hat{X} \in \VSLAM_n(3)$ with kinematics \eqref{eq:group_observer_state}.
Assume that $ U $ is bounded, and that $y_i^\times y_i^\times V_U$ is persistently exciting, in the sense that there exist $T > 0$ and $\mu > 0$ such that
\begin{align}
    \frac{1}{T} \int_{t}^{t+T} \Vert y_i^\times y_i^\times V_U \Vert \td \tau \geq \mu,
    \label{eq:pe_velocity}
\end{align}
for each $i$ and all $t > 0$.
Assume that there exist bounds $0 < \ub{r} < \ob{r} \in \R$ such that
\[
\ub{r} \leq \vert |p_i - x_P| \vert \leq \ob{r}
\]
for all time.

Then the landmark correction terms \eqref{eq:landmark_corrections}
define an almost semi-globally asymptotically stabilising (Def.~\ref{def:semiGAS}) correction term for the error dynamics of $d_i, \tilde{r}_i$ \textup{(\ref{eq:output_error},\ref{eq:range_error})} around the equilibrium $(\mr{y}_i, 1)$ with exception set
\begin{align}
\chi = \{ (d, \tilde{r})_i \in (\Sph^2 \times \R_+)^n \; \vert\; d_i &= -\mr{y}_i \text{ for some }
i \in [1, \ldots, n]\}.
\label{eq:chi}
\end{align}
Moreover, as $(d_i, \tilde{r}_i) \to (\mr{y}_i, 1)$, the estimated state $\hat{\xi} \to \xi$ converges to the true state up to the SLAM manifold equivalence.
\end{theorem}

\begin{proof}
The outline of the proof is as follows.
We begin by choosing the parameters of the correction terms to depend on the initial conditions.
Next, it is shown that the observer equations are well-defined for all time.
We introduce storage functions in \eqref{eq:landmark_lyap} and proceed to show that they are non-increasing over time in \eqref{eq:l_dot}.
Then we apply Barbalat's lemma along with persistence of excitation to show that the storage functions converge to zero over time, leading to \eqref{eq:rinv_convergence} and \eqref{eq:rdiff_times_yyV}.
Finally, we show that convergence of the error dynamics is equivalent to convergence of the true and estimated states on the SLAM manifold.

Let $K$ be a compact set in the complement of $\chi$ \eqref{eq:chi}.
Then there exists $\tilde{r}^m > 0$ such that $(d_i, \tilde{r}_i) \in K$ implies that $\tilde{r}_i \geq \tilde{r}^m$.
Choose $k_0 = \frac{1}{2} \tilde{r}^m$ and assume that $(d_i(0), \tilde{r}_i(0)) \in K$.
Then
\begin{align*}
    \tilde{r}_i(0) &\geq k_0, &
    \hat{r}_i(0) &\geq k_0 r_i(0) \geq k_0 \ub{r} > 0,
\end{align*}
for each $i$.
Choose $\ub{c} = \ub{r}$ and $\epsilon = \min(k_0, \frac{1}{2}) \ub{r}$.
Then $\hat{r}_i(0) > \epsilon > 0$ for every $i$.
By continuity of the solutions there exist $T_i > 0$ such that $\hat{r}_i(t) > \epsilon$ and $1 + d_i^\top \mr{y}_i > 0$, and hence $\beta^\epsilon_{\ub{r}}(\hat{r}_i)$ and $\Gamma_i$ are well-defined, for $t \in [0,T_i)$.
(We will show later that $T_i$ can be chosen arbitrarily large, and that $\hat{r}_i(t) > \epsilon$ and $1 + d_i^\top \mr{y}_i > 0$ hold for all time.)

By definition, $r_i = a_i^{-1} \mr{r}_i$, and $\hat{r}_i = \hat{a}_i^{-1} \mr{r}_i$.
Differentiating these with respect to time yields
\begin{align}
\dot{r}_i &= -V_U^\top y_i, \notag \\
&= - d_i^\top \hat{Q}_i V_U, \notag \\
\dot{\hat{r}}_i &= -V_U^\top \hat{y}_i + \hat{r}_i \gamma_i, \notag \\
&= -V_U^\top y_i
+ \alpha_i \beta^\epsilon_{\ub{r}} (\hat{r}_i)
\label{eq:rhat_dot} 
+ \frac{\alpha_i}{\hat{r}_i}\left((1-d_i^\top \mr{y}_i) d_i^\top \hat{Q}_i V_U - {\mr{y}}_i^{\top} (d_i \times \hat{Q}_i V_U)^\times d_i \right),
\end{align}
where $\hat{y}_i = h^i(\hat{\xi})$ is the estimated measurement.
Since $\beta^\epsilon_{\ub{r}}$ \eqref{eq:beta_function} is a barrier function ensuring that $\hat{r}_i(t) > \epsilon$  ($\beta^\epsilon_{\ub{r}} (\hat{r}_i) \to \infty$ as $\hat{r}_i \searrow \epsilon$) and the remaining terms in \eqref{eq:rhat_dot} are bounded, it follows that $\hat{r}_i$ is well defined $\forall t \in [0,T_i)$.

Differentiating the output error $d_i$ yields
\begin{align*}
\dot{d}_i &= \frac{\td}{\td t} \rho(X\hat{X}^{-1}, \mr{y})_i
= \frac{\td}{\td t} \hat{Q}_i Q_i^\top \mr{y}_i, \\
&= \left( \hat{Q}_i \Lambda_Q (U, R_{\hat{P}}^\top(\hat{p}_i - x_{\hat{P}})) - \Gamma_i\hat{Q}_i \right) Q_i^\top \mr{y}_i \\
&\hspace{1cm} - \hat{Q}_i \Lambda_Q (U, R_P^\top(p_i - x_P)) Q_i^\top \mr{y}_i.
\end{align*}
Observe that $R_{\hat{P}}^\top(\hat{p}_i - x_{\hat{P}}) = \hat{r}_i \hat{y}_i$ and $R_P^\top(p_i - x_P) = r_i y_i$.
Using this, the derivative of $d_i$ is simplified to
\begin{align}
    \dot{d}_i &= \Ad_{\hat{Q}_i} \left( \Lambda_Q (U, \hat{r}_i \hat{y}_i) - \Lambda_Q (U, r_i y_i) \right) d_i - \Gamma_i d_i, \notag \\
    &= \Ad_{\hat{Q}_i} \left(\frac{\hat{y}_i \times V_U}{\hat{r}_i} - \frac{y_i \times V_U}{r_i}\right)^\times d_i - \Gamma_i d_i, \notag \\
    &= \left(\frac{{\mr{y}_i}^\times \hat{Q}_i V_U}{\hat{r}_i} - \frac{d_i^\times \hat{Q}_i V_U}{r_i}\right)^\times d_i - \Gamma_i d_i, \label{eq:delta_dot}
\end{align}
where the last step is obtained by using various identities involving the skew symmetric operator.

Recalling \eqref{eq:delta_dot} and using identities involving the skew symmetric operator and projector, the derivative of $1 - {\mr{y}}_i^\top d_i$ is
\begin{align}
    &\ddt \left( 1 - {\mr{y}}_i^\top d_i \right)
    = - {\mr{y}}_i^\top \dot{d}_i, \notag \\
    &= - {\mr{y}}_i^\top \left(\frac{{\mr{y}_i}^\times \hat{Q}_i V_U}{\hat{r}_i} - \frac{d_i^\times \hat{Q}_i V_U}{r_i}\right)^\times d_i + {\mr{y}}_i^\top \Gamma_i d_i, \notag \\
    &= - {\mr{y}}_i^\top \left(\frac{d^\times \hat{Q}_i V_U}{\hat{r}_i} - \frac{d_i^\times \hat{Q}_i V_U}{r_i}\right)^\times d_i
    + \left(\frac{d_i^\top\hat{Q}_i V_U}{\hat{r}_i(1 + d_i^\top \mr{y}_i)} -\frac{k_i}{(1+ d_i^\top \mr{y}_i)^2} \right) {\mr{y}}_i^\top (d_i^\times \mr{y}_i)^\times d_i, \notag \\
    &= (r_i^{-1} - {\hat{r}_i}^{-1}){\mr{y}}_i^\top \left(d^\times \hat{Q}_i V_U \right)^\times d_i
    + \left(\frac{d_i^\top\hat{Q}_i V_U}{\hat{r}_i(1 + d_i^\top \mr{y}_i)} -\frac{k_i}{(1+ d_i^\top \mr{y}_i)^2} \right) {\mr{y}_i}^\top \Pi_{d_i} \mr{y}_i, \notag \\
    &= (r_i^{-1} - {\hat{r}_i}^{-1}){\mr{y}}_i^\top \left(d^\times \hat{Q}_i V_U \right)^\times d_i
    + \left(\frac{d_i^\top\hat{Q}_i V_U}{\hat{r}_i} - \frac{k_i}{1+ d_i^\top \mr{y}_i} \right) (1 - d_i^\top \mr{y}_i). \label{eq:bearing_error_dynamics}
\end{align}
Observe that $k_i \frac{1-d_i^\top \mr{y}_i}{1+d_i^\top \mr{y}_i} \to \infty$ as $1-d_i^\top \mr{y}_i \nearrow 2$.
Since all other terms in \eqref{eq:bearing_error_dynamics} are bounded, it follows that $1-d_i^\top \mr{y}_i < 2 - \nu$ for some small $\nu > 0$, and hence $1+d_i^\top \mr{y}_i > \nu > 0$ and $\Gamma_i$ is well-defined and bounded for all $t \in [0, T_i)$.

We show by contradiction that the domain of definition $[0, T_i)$ can be extended arbitrarily.
Suppose, for some $i$, that $T_i'$ is the largest value such that $\Gamma_i$ and $\beta^\epsilon_{\ub{r}} (\hat{r}_i)$ are well-defined.
Then $\Gamma_i$ and $\beta^\epsilon_{\ub{r}} (\hat{r}_i)$ are both bounded on $[0,T_i')$ by the arguments above, and continuous.
It follows that their limits as $t \to T_i'$ exist and are finite.
But then, by continuity of solutions, $\Gamma_i$ and $\beta^\epsilon_{\ub{r}} (\hat{r}_i)$ can be extended to $[0, T_i' + t_\Delta)$ for some sufficiently small $t_\Delta > 0$.
This contradicts the assumption that $T_i'$ is the maximum value for which $\Gamma_i$ and $\beta^\epsilon_{\ub{r}} (\hat{r}_i)$ are well-defined on $[0,T_i')$, and therefore no such maximum value can exist.
Hence $\Gamma_i$ and $\beta^\epsilon_{\ub{r}} (\hat{r}_i)$ are well-defined on $[0, \infty)$, and so are the observer dynamics.

For each $i$, define the storage function
\begin{align}
l_i(d_i,\tilde{r}_i;r_i) &:= \frac{r_i}{2} \Vert {\mr{y}_i} - d_i \Vert^2  + \frac{r_i^2}{2\alpha_i} (1 - \tilde{r}_i)^2 \notag \\
&=
r_i (1 - {\mr{y}}_i^\top d_i)+ \frac{1}{2\alpha_i} (r_i - \hat{r}_i)^2,
\label{eq:landmark_lyap}
\end{align}
in the variables $(d_i,\tilde{r}_i)$ where the second line follows from $\Vert d_i \Vert = \Vert \mr{y}_i \Vert = 1$ and substituting for $\tilde{r}_i$.
Note that $l_i$ depends on the time varying range $r_i$ as a parameter.
By assumption, $r_i \geq \ub{r} > 0$ and the storage functions $l_i$ are positive definite.

The derivative of $l_i$ may now be computed as follows:
\begin{align}
    \dot{l}_i
    &=
    \dot{r}_i (1 - {\mr{y}}_i^\top d_i)
    + r_i \ddt (1 - {\mr{y}}_i^\top d_i)
    \notag \\ &\phantom{=} \hspace{0.0cm}
    + \frac{1}{\alpha_i} (r_i - \hat{r}_i) (\dot{r}_i - \dot{\hat{r}}_i), \notag \\
    &=
    - (1 - {\mr{y}}_i^\top d_i) d_i^\top \hat{Q}_i V_U
    \notag \\ &\phantom{=} \hspace{0.0cm}
    + r_i (r_i^{-1} - {\hat{r}_i}^{-1}){\mr{y}}_i^\top \left(d^\times \hat{Q}_i V_U \right)^\times d_i
    \notag \\ &\phantom{=} \hspace{0.0cm}
    + r_i \left(\frac{d_i^\top\hat{Q}_i V_U}{\hat{r}_i} - \frac{k_i}{1 + d_i^\top \mr{y}_i} \right) (1 - d_i^\top \mr{y}_i)
    \notag \\ &\phantom{=} \hspace{0.0cm}
    +  (\hat{r}_i - r_i) \beta^\epsilon_{\ub{r}} (\hat{r}_i)
    \notag \\ &\phantom{=} \hspace{0.0cm}
    + (\hat{r}_i - r_i) \hat{r}_i^{-1} (1-d_i^\top \mr{y}_i) d_i^\top \hat{Q}_i V_U
    \notag \\ &\phantom{=} \hspace{0.0cm}
    - (\hat{r}_i - r_i) \hat{r}_i^{-1} {\mr{y}}_i^{\top} (d_i \times \hat{Q}_i V_U)^\times d_i, \notag \\
    &=
    r_i (r_i^{-1} - {\hat{r}_i}^{-1}){\mr{y}}_i^\top \left(d^\times \hat{Q}_i V_U \right)^\times d_i
    \notag \\ &\phantom{=} \hspace{0.0cm}
    - k_i r_i \frac{1 - d_i^\top \mr{y}_i}{1 + d_i^\top \mr{y}_i}
    +  (\hat{r}_i - r_i) \beta^\epsilon_{\ub{r}} (\hat{r}_i)
    \notag \\ &\phantom{=} \hspace{0.0cm}
    - (\hat{r}_i - r_i) \hat{r}_i^{-1} {\mr{y}}_i^{\top} (d_i \times \hat{Q}_i V_U)^\times d_i, \notag \\
    &=
    - k_i r_i \frac{1 - d_i^\top \mr{y}_i}{1 + d_i^\top \mr{y}_i}
    +  (\hat{r}_i - r_i) \beta^\epsilon_{\ub{r}} (\hat{r}_i).
    \label{eq:l_dot}
\end{align}
The second term is negative semi-definite, since it is zero when $\hat{r}_i \geq \ub{r}$ \eqref{eq:beta_function}, and otherwise $\hat{r}_i < \ub{r} \leq r_i$ and $\hat{r}_i - r_i \leq 0$.
It follows that $\dot{l}_i (t) \leq 0$, and $l_i(t) \leq l_i(0)$ for all $t$.

Barbalat's Lemma \cite[Lemma 4.2]{1991_slotine_nonlinear_control} is used to prove $\dot{l}_i \to 0$.
To show that $\dot{l}_i$ is uniformly continuous, it is sufficient to show that $\ddot{l}_i$ is bounded.
Recall that $r_i$ is bounded above and below by assumption, and that $\hat{r}_i(t) > \epsilon$ for all time.
Computing the second derivative of $l_i$, one has
\begin{align*}
    \ddot{l}_i
    &= - k_i \dot{r}_i\frac{1 - {\mr{y}_i}^\top d_i}{1 + {\mr{y}_i}^\top d_i}
    - \frac{2 r_i {\mr{y}_i}^\top \dot{d}_i}{(1 + {\mr{y}_i}^\top d_i)^2}
    +(\hat{r}_i - r_i) \frac{\partial \beta^\epsilon_{\ub{r}}}{\partial \hat{r}_i} \dot{\hat{r}}_i.
\end{align*}
It suffices to show that the component terms are all bounded.
In the first and second terms, $1/(1+{\mr{y}_i}^\top d_i)$ is bounded from the above discussion.
The components $ r_i {\mr{y}_i}^\top \dot{d}_i$  and $\dot{r}_i(1 - {\mr{y}_i}^\top d_i)$ are also bounded due to the boundedness of $\hat{r}_i$, $r_i$, and the assumption that $U$ is bounded.
As for the third term, the component $\dot{\hat{r}}_i$ is bounded since the velocity input $V_U$ is bounded and $\hat{r}_i$ is lower bounded.
This means that both $\beta^\epsilon_{\ub{r}}(\hat{r}_i)$ and its derivative with respect to $\hat{r}_i$ are bounded.
Therefore, $\dot{l}_i$ is uniformly continuous, and by Barbalat's lemma, $\dot{l}_i \to 0$.
This implies that $d_i \to \mr{y}_i$, since both terms that appear in $\dot{l}_i$ \eqref{eq:l_dot} are non-positive .

It remains to show $\tilde{r}_i \to 1$ (or equivalently $\hat{r}_i \to r_i$).
Applying Barbalat's Lemma to $d_i(t)$ and exploiting the same bounding arguments as before, it is straightforward to verify that $\dot{d}_i \to 0$.
It is also easily verified that $\Gamma_i \to 0$ as $d_i \to \mr{y}_i$.

From \eqref{eq:delta_dot}, $\dot{d}_i$ may be written as
\begin{align*}
    \dot{d}_i
    &= \left(\frac{{\mr{y}_i}^\times \hat{Q}_i V_U}{\hat{r}_i} - \frac{d_i^\times \hat{Q}_i V_U}{r_i}\right)^\times d_i - \Gamma_i d_i, \\
    &= \frac{d_i^\times d_i^\times \hat{Q}_i V_U}{r_i} - \frac{d_i^\times {\mr{y}_i}^\times \hat{Q}_i V_U}{\hat{r}_i} - \Gamma_i d_i, \\
    &= (r_i^{-1} - \hat{r}_i^{-1})d_i^\times d_i^\times \hat{Q}_i V_U + \frac{d_i^\times (d_i - \mr{y}_i)^\times \hat{Q}_i V_U}{\hat{r}_i} - \Gamma_i d_i, \\
    &= (r_i^{-1} - \hat{r}_i^{-1}) \hat{Q}_i y_i^\times y_i^\times V_U + \frac{d_i^\times (d_i - \mr{y}_i)^\times \hat{Q}_i V_U}{\hat{r}_i} - \Gamma_i d_i.
\end{align*}
Clearly, $\hat{r}_i^{-1} d_i^\times (d_i - \mr{y}_i)^\times \hat{Q}_i V_U \to 0$ as $d_i \to \mr{y}_i$.
Hence
\begin{align}
    \dot{d}
    &\to (r_i^{-1} - \hat{r}_i^{-1}) \left(\hat{Q}_i {y_i}^\times {y_i}^\times V_U\right),
    \label{eq:rinv_convergence}
\end{align}
as $d_i \to \mr{y}_i$.
Therefore, $(\hat{r}_i - r_i) \frac{1}{\hat{r}_i r_i}\Vert {y_i}^\times {y_i}^\times V_U \Vert \to 0$.
Recall that $r_i \leq \ob{r}$. The estimated range $\hat{r}_i$ is also bounded above by a constant $\ob{\hat{r}}_i$ depending on the initial value of $l_i(0)$.
Since $\vert r_i - \hat{r}_i \vert < (\ob{r} \ob{\hat{r}}_i) \vert r_i^{-1} - \hat{r}_i^{-1} \vert$, it follows from \eqref{eq:rinv_convergence} that
\begin{align} \label{eq:rdiff_times_yyV}
    (r_i - \hat{r}_i) \Vert {y_i}^\times {y_i}^\times V_U \Vert \to 0.
\end{align}
Each $l_i$ must converge to a positive constant $c^0_i \leq l_i(0)$ as $\dot{l}_i \to 0$.
Therefore, \eqref{eq:landmark_lyap} ensures that $(\hat{r}_i - r_i)  \to \pm \sqrt{2 \alpha_i c^0_i}$.
Integrating \eqref{eq:rdiff_times_yyV} over a period of time $T$, and using the fact that $(r_i - \hat{r}_i)$ is converging to a constant, it follows that
\[
    (r_i - \hat{r}_i)\int_{t}^{t+T}\Vert {y_i}^\times {y_i}^\times V_U \Vert d\tau \to 0.
\]
Using the persistence of excitation assumption \eqref{eq:pe_velocity}, it must be that $\hat{r}_i \to r_i$.

Recall the exception set $\chi$ \eqref{eq:chi}.
It is straightforward to verify that this set has measure zero.
Since the initial choice of compact set $K$ in the complement of $\chi$ was arbitrary, then the equilibrium $(\mr{y}_i, 1)$ of $(d_i, \tilde{r}_i)$ is almost semi-globally asymptotically stabilisable by the proposed correction terms (Def.~\ref{def:semiGAS}).


Observe that, at the equilibrium,
\begin{gather*}
    \hat{y}_i = \rho(\hat{X}, \mr{y}_i) = \rho(\hat{X}, d_i) = y_i, \\
    \hat{r}_i = \tilde{r}_i r_i = r_i.
\end{gather*}
Using this, the ego-centric coordinates of the SLAM configuration and their estimates satisfy
\begin{align*}
R_P^\top (p_i - x_P) &= r_i y_i = \hat{r}_i \hat{y}_i = R_{\hat{P}}^\top (\hat{p}_i - x_{\hat{P}}),
\end{align*}
and thus
\[
p_i = (R_{\hat{P} P^{-1}})^\top (\hat{p}_i - x_{\hat{P} P^{-1}}).
\]
Therefore, at the equilibrium point, each $p_i$ is related to each $\hat{p}_i$ by the same rigid body transformation $S = \hat{P} P^{-1}$.
Moreover, it is clear that $P = S^{-1} \hat{P}$.
Hence, the two configurations on $\mr{\calT}_n(3)$ are equivalent on the SLAM manifold, $\xi \simeq \hat{\xi}$.
This completes the proof.
\end{proof}

\subsection{Total Space Representative}\label{sec:Delta}
In Theorem \ref{thm:landmark_observer} the convergence result is independent of the choice of correction term $\Delta$.
This is due to the ego-centric nature of the group action considered and the invariance properties of the SLAM manifold, which cause the inertial frame of a SLAM system to be unobservable.
In essence, choosing $\Delta$ will influence the element $S \in \SE(3)$ that relates the reference and estimated states, but does not influence the convergence of the SLAM error.
Nevertheless, it is clear that as the error converges, it is desirable that the $S(t)$ that relates the reference and established states converges to a constant, essentially capturing the ``inertial map'' property that is desired in visual odometry.
It is a key contribution of this paper to observe that imposing this constraint is a separate requirement from the underlying SLAM solution, that is, we must introduce an additional criterion that captures this property and then use this to design the correction term $\Delta$.

The criterion that we propose to minimize is the weighted mean velocity of the landmark points
\[
\sum_{i=1}^n \kappa_i  \Vert \dot{\hat{p}}_i \Vert^2
\]
For a static environment, the true landmark points are not moving.
For the observer estimate these points may be moving, due to residue velocity associated with the landmark error correction, but also importantly, due to the correction term $\Delta$ that is moving the entire SLAM configuration.
The motion due to $\Delta$ will be strongly correlated, while it is expected that the residue velocity due to the correction terms will be uncorrelated and for large constellations of points average to zero.
Choosing $\Delta$ to minimize this additional criteria can be thought of as minimizing instantaneous map drift.


\begin{proposition} \label{prop:robot_correction}
Let the origin configuration $\mr{\xi} = (\mr{P}, \mr{p}) \in \mr{\calT}_n(3)$, the observer state $\hat{X} = (\hat{A}, (\hat{Q}, \hat{a})_i) \in \VSLAM_n(3)$, and the correction term $\Delta_{\hat{X}} = (\Delta, (\Gamma, \gamma)_i)$ be defined as in the statement of Theorem \ref{thm:landmark_observer}.
Let $\hat{\xi} = (\hat{P}, \hat{p}_i) = \Upsilon(\hat{X}, \mr{\xi}) \in \mr{\calT}_n(3)$ be the estimated state on the total space, and let $\hat{q}_i = R_{\hat{P}}^\top (\hat{p}_i - x_{\hat{P}})$ for each $i$.
Then the solution to
\begin{align} \label{eq:lsq_problem}
\Delta = \argmin_{\Delta \in \se(3)} \left\{ \sum_{i=1}^n \kappa_i \Vert  \dot{\hat{p}}_i \Vert^2 \right\},
\end{align}
where $\kappa_i$ are positive scalars, is given by
\begin{align}
\Delta = \Ad_{\hat{A}}(\Omega_\Delta^\times, V_\Delta),
\label{eq:Delta_correction}
\end{align}
where $\Omega_\Delta$ and $V_\Delta$ are determined by
\begin{align} \label{eq:lsq_solution}
\begin{pmatrix}
\Omega_\Delta \\ V_\Delta
\end{pmatrix}
=
- \left(
\sum_{i=1}^n \kappa_i \begin{pmatrix}
    \hat{q}_i^\times \hat{q}_i^\times & \hat{q}_i^\times \\
    - \hat{q}_i^\times & I_3
\end{pmatrix}
\right)^{-1} 
\left(
\sum_{i=1}^n \kappa_i \begin{pmatrix}
    \hat{q}_i^\times \Ad_{\hat{Q}_i^\top}(\Gamma_i)\hat{q}_i \\
    \gamma_i \hat{q}_i + \Ad_{\hat{Q}_i^\top}(\Gamma_i)\hat{q}_i
\end{pmatrix}
\right),
\end{align}
so long as the inverse in \eqref{eq:lsq_solution} remains well-defined.
\end{proposition}
\begin{proof}
First, observe that
\begin{align*}
\hat{q}_i
&= R_{\hat{P}}^\top (\hat{p}_i - x_{\hat{P}}), \\
&= R_{\hat{P}}^\top (\hat{a}_i^{-1}R_{\hat{P}} \hat{Q}_i^\top R_{\mr{P}}^\top(\mr{p}_i - x_{\mr{P}}) + x_{\hat{P}} - x_{\hat{P}}), \\
&= \hat{a}_i^{-1} \hat{Q}_i^\top R_{\mr{P}}^\top(\mr{p}_i - x_{\mr{P}}).
\end{align*}
Equation \eqref{eq:lsq_problem} presents a weighted least squares problem, and to solve it we analyse the component expressions $\kappa_i\dot{\hat{p}}_i$.
The time derivative of each $\hat{p}_i$ needs to be computed.
Recall that the velocity lift $\Lambda$ is defined precisely so that $\dot{\hat{p}}_i = 0$ when the correction terms are set to zero.
Since differentiation is a linear operation, this means that
\begin{align*}
\dot{\hat{p}}_i
&= \frac{\td}{\td t} \left( R_{\mr{P} \hat{A}} \hat{a}_i^{-1} \hat{Q}_i^\top R_{\mr{P}}^\top(\mr{p}_i - x_{\mr{P}}) + x_{\mr{P} \hat{A}} \right) \\
&= - R_{\mr{P} \hat{A}} \Omega_\Delta^\times \hat{a}_i^{-1} \hat{Q}_i^\top R_{\mr{P}}^\top(\mr{p}_i - x_{\mr{P}}) - R_{\mr{P} \hat{A}} V_\Delta \\
&\hspace{1cm} + \gamma_i R_{\mr{P} \hat{A}} \hat{a}_i^{-1} \hat{Q}_i^\top R_{\mr{P}}^\top(\mr{p}_i - x_{\mr{P}}) \\
&\hspace{1cm} + R_{\mr{P} \hat{A}} \hat{a}_i^{-1} \hat{Q}_i^\top \Gamma_i R_{\mr{P}}^\top(\mr{p}_i - x_{\mr{P}}), \\
&= - R_{\mr{P} \hat{A}} \Omega_\Delta^\times \hat{q}_i - R_{\mr{P} \hat{A}} V_\Delta + \gamma_i R_{\mr{P} \hat{A}} \hat{q}_i \\
&\hspace{1cm} + R_{\mr{P} \hat{A}} \Ad_{\hat{Q}_i^\top}(\Gamma_i) \hat{q}_i, \\
\Vert \dot{\hat{p}}_i \Vert
&= \Vert - \Omega_\Delta^\times \hat{q}_i - V_\Delta + \gamma_i \hat{q}_i + \Ad_{\hat{Q}_i^\top}(\Gamma_i) \hat{q}_i \Vert, \\
&= \left\Vert \begin{pmatrix}
-\hat{q}_i^\times & I_3
\end{pmatrix} \begin{pmatrix}
\Omega_\Delta \\ V_\Delta
\end{pmatrix} - \begin{pmatrix}
\gamma_i \hat{q}_i + \Ad_{\hat{Q}_i^\top}(\Gamma_i) \hat{q}_i
\end{pmatrix} \right\Vert.
\end{align*}
Therefore, by the theory of Weighted Least Squares, \eqref{eq:lsq_solution} is exactly the solution to \eqref{eq:lsq_problem}, as required.
\end{proof}
Proposition \ref{prop:robot_correction} provides a clear way to choose a correction term based on the static landmark assumption, and allows for scalars $\kappa_i$ to be chosen to weight the optimisation.
The computational and memory costs of the correction terms scale linearly with the number of landmarks as opposed to alternative observer designs \cite{2015_Manuel_vslam,2012_Strasdat_CVIU} which scale quadratically.

\section{Simulation Results} \label{sec:simulations}
To verify the landmark observer design in Theorem \ref{thm:landmark_observer} and the robot correction term in Proposition \ref{prop:robot_correction}, we conducted a simulation of a flying vehicle equipped with a monocular camera, observing 5 stationary landmarks as it moves in a circular trajectory with a constant body-fixed velocity $U = (\Omega_U^\times, V_U)$, where $\Omega_U = (0,0,0.5)$rad/s and $V_U = (1.5,0,0)$m/s.
For simplicity, it is assumed that the camera frame coincides with the body-fixed frame of the vehicle.
The initial position of the vehicle was set to $(3,3,5)$m with its rotational axes aligned with the inertial frame $\{0\}$.
The positions of the landmarks were initialised to random positions $(p^1_i, p^2_i, 0)$ on the ground plane with $p^1_i, p^2_i \sim N(0, 5^2)$m.

The origin position $\mr{P}$ of the robot is set to the identity $I_4$, and the origin landmark positions $\mr{p}_i$ are set to $10 y_i(0)$, where $y_i(0)$ are the measured bearings to the true landmark positions at time 0.
That is, the estimated points are initialised with correct bearings and an arbitrary depth of 10 m.
The observer is defined on $\VSLAM_5(3)$ with kinematics given by \eqref{eq:group_observer_state}, landmark correction terms $(\Gamma_i, \gamma_i)$ given by \eqref{eq:landmark_corrections}, where $k_i = 5$ and $\alpha_i = 500$, and robot correction term $\Delta$ given as in Proposition \ref{prop:robot_correction} with $\kappa_i = 1$.
At the end of the simulation, the estimated system state is aligned with the true system state by matching the true and estimated robot poses.
Figure \ref{fig:delta_trajectories} shows the trajectories of estimated landmark positions and robot position over time, as well as the true landmark positions and the true robot trajectory, and figure \ref{fig:delta_lyap} shows the evolution of each of the landmarks' associated storage functions, as defined in \eqref{eq:landmark_lyap}.

This simulation provides a simple demonstration of performance of the proposed observer and illustrates typical trajectories of the landmark estimates during a repeating motion such as the circle.
The estimated landmark positions can be seen to converge to the true landmark positions in a natural manner.
The choice to initialise landmarks as having a bearing matching the initial measurement is a natural one for practical implementation of the algorithm, although Theorem \ref{thm:landmark_observer} provides that almost any initial conditions will converge.
This almost semi-global convergence is a key property of the observer presented here that is not available in many of the state-of-the-art solutions.

\begin{figure}[!htb] \centering
\includegraphics[width=0.9\linewidth]{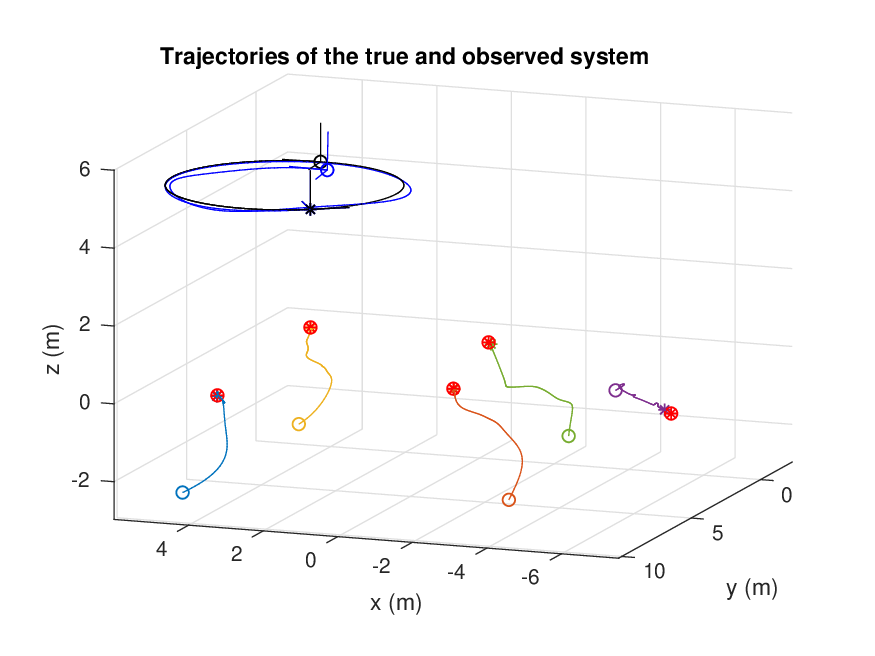}
\caption{The simulated trajectories of observer landmark estimates from initial positions with correct bearings but fixed depth of 10m.
The observer landmark trajectories are shown in a range of colours matching those used in Figure \ref{fig:delta_lyap}, and the true landmark positions are shown in red.
The observer robot trajectory is shown in blue, and the true robot trajectory is shown in black.
The ($\circ$) and ($\star$) markers, respectively, denote the start and end of the trajectories of all the objects shown.
}
\label{fig:delta_trajectories}
\end{figure}

\begin{figure}[!htb] \centering
\includegraphics[width=0.85\linewidth]{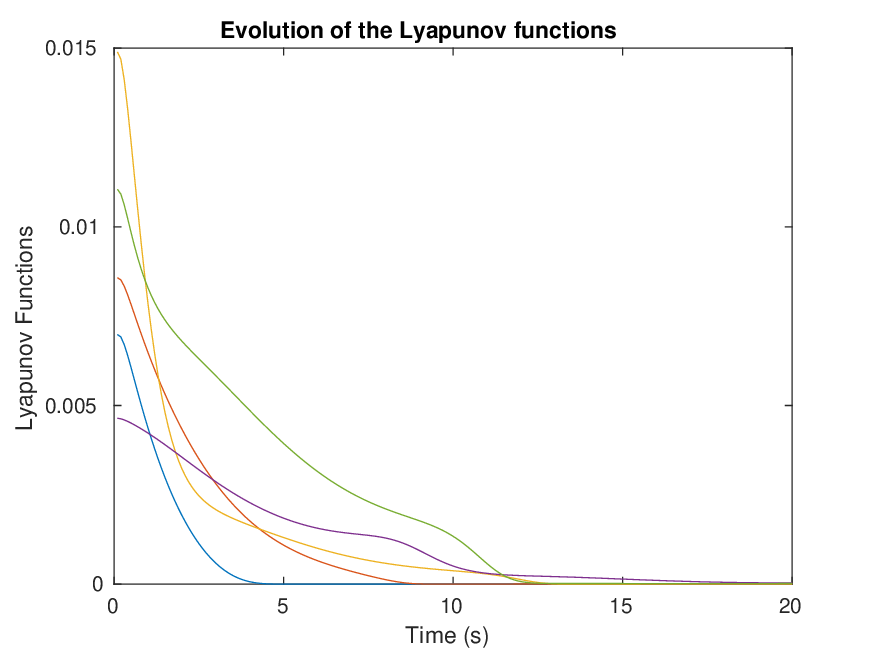}
\caption{The evolution of the storage functions of each of the landmarks in the system shown in Figure \ref{fig:delta_trajectories}.
The colours match the colours of the landmark trajectories in Figure \ref{fig:delta_trajectories}.
The initial convergence of the landmarks is quick as the bearings converge, and then slows as the depths gradually converge.}
\label{fig:delta_lyap}
\end{figure}

\section{Experimental Results} \label{sec:experiments}
To demonstrate the observer described in Theorem \ref{thm:landmark_observer} in a real-world scenario, we gathered video, GPS, and IMU data from a Disco Parrot fixed-wing UAV flying outdoors.
Image features were identified using OpenCV's goodFeaturesToTrack, and subsequently tracked using OpenCV's calcOpticalFlowPyrLK.
These image features were then corrected for camera intrinsics and converted to spherical bearing coordinates before being used as landmark inputs to the observer.
Landmarks are added to the system state after being observed for two frames so that their depths can be initialised from optical flow.
When a landmark is no longer visible, it is removed from the observer state.

Initially, the input velocities $U = (\Omega_U, V_U)$ to the system were estimated by combining the GPS signal (to obtain scale information) with egomotion estimated using the IMU and optical flow from the video stream, as outlined in \cite{2011_schill_egomotion}.
Once sufficiently many landmarks are initialised, the optical flow vectors of each of the landmarks were combined with the existing landmark estimates to compute the input velocity $U = (\Omega_U, V_U)$.
The observer was implemented using Euler integration with gain parameters set to $k_i = 5.0$ and $\alpha_i = 0.5$ for each $i$.
The video recorded had a frame rate of $30$ fps, leading to the Euler integration step being set to $\td t = 0.033$ s.
GPS data was recorded at $25$ Hz in order to compare with the observer's estimated trajectory.
The observer trajectory was aligned to the GPS trajectory using the Umeyama method \cite{1991_umeyama_TPAMI}.
Figure \ref{fig:experiment-aligned-trajectories} shows the aligned trajectories according to the observer and according to the GPS in the $x$ and $y$ directions, where the $z$ direction refers to the plane's altitude.
Figure \ref{fig:experiment-map} shows the final positions of all landmark points in addition to the observer- and GPS-estimated trajectories.
Figure \ref{fig:experiment-frame} shows a frame taken from the video stream used in the experiment, with lines to represent the optical flow tracking overlaid.
A video showcasing the feature tracking system is available online\footnote{\url{https://www.youtube.com/watch?v=QzIxh2eM1_s}}.
The quality of the trajectory and map provided in Figure \ref{fig:experiment-map} show the robustness of the observer to noisy bearing measurements in practice.

\begin{figure}[!htb] \centering
\includegraphics[width=0.85\linewidth]{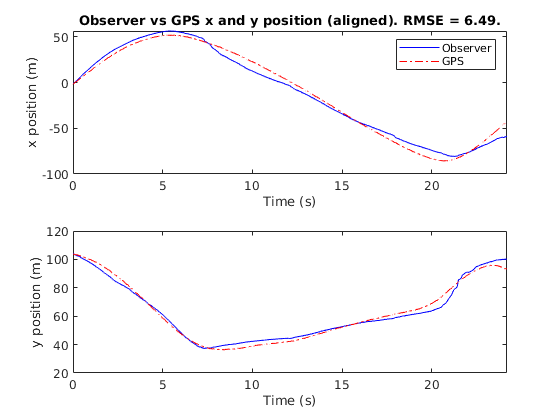}
\caption{The x and y positions of the UAV according to the aligned Observer (blue) and GPS (red).}
\label{fig:experiment-aligned-trajectories}
\end{figure}

\begin{figure}[!htb] \centering
\includegraphics[width=0.8\linewidth]{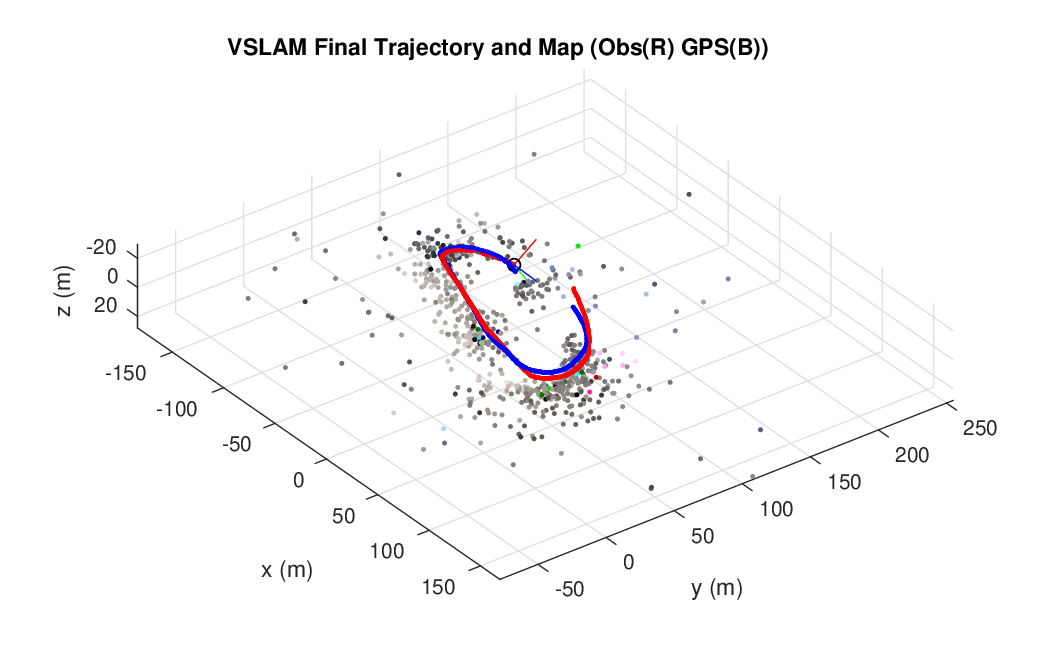}
\caption{The full trajectory of the UAV according to the aligned Observer (blue) and GPS (red), and the final positions of all of the landmarks, coloured with the colour of the pixel where they were first observed.}
\label{fig:experiment-map}
\end{figure}

\begin{figure}[!htb] \centering
    \includegraphics[width=0.8\linewidth]{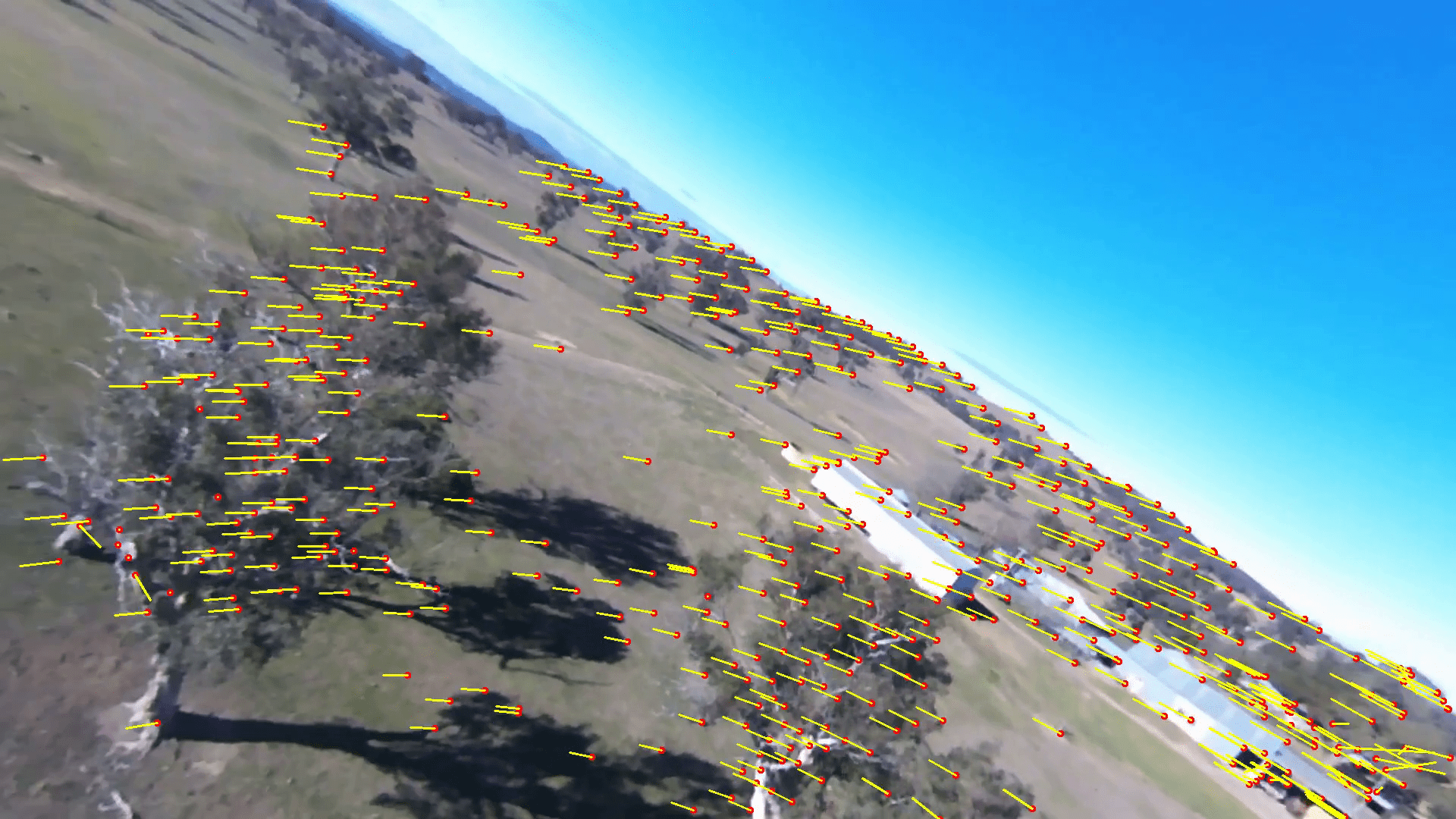}
    \caption{A single frame of the video stream used in the experiment.
    Red circles represent the image features being tracked, and yellow lines represent the vector of motion of the image features between the current frame and the previous frame.
    }
    \label{fig:experiment-frame}
\end{figure}

\section{Conclusion}
This paper presents an observer design posed on a symmetry group for the VSLAM problem.
The SLAM manifold introduced in \cite{2017_Mahony_cdc} and the symmetry group discussed in \cite{2019_vangoor_cdc_vslam} are reintroduced and exploited in the observer design.
The observer is formulated on output errors, and provides a clear way to change the gains for bearing and depth of landmarks separately.
The almost semi-global convergence of the proposed observer improves on the properties of state-of-the-art Extended Kalman Filter systems, which suffer from linearisation errors.
While research into the development of non-linear observers for the SLAM problem is only recent, the observer for VSLAM presented in this paper demonstrates some of the key advantages the approach can offer.

\appendix
\section{Almost Semi-Globally Asymptotically Stabilising Controls and Corrections}

The concept of semi-global asymptotic stabilisability was introduced in \cite{1994_Teel_SCL} to model the dependence of gain on the basin of attraction in feedback stabilisation of a dynamical system.
In the context of an observer analysis, this definition can be transferred to stabilisability of the error dynamics by correction.
However, the classical concept introduced by Teel and Praly does not capture topological constraints associated with stability analysis on manifolds.
For a large class of manifolds, including Lie-groups with $\SO(3)$ as a subgroup, these topological constraints prevent the existence of globally smooth asymptotically stable error dynamics \cite{2000_Bhata_SCL}.
On such spaces, smooth error dynamics will always admit an exception set $\chi$ of unstable or hyperbolic critical points that cannot be part of the basin of attraction of the desired equilibrium.

We consider a system-observer pair, where the observer has the internal model principle, coupled through a correction function $\Delta(\hat{x},y)$ depending on the observer state and system output.
We assume that there is a well defined error function $e : \grpG \times \calM \to \calM$ where $\grpG$ is the observer state space and $\calM$ is the system state space.
The error dynamics evolve on the manifold $\calM$ depending on the system and observer state evolution as well as any exogenous inputs such as velocities.

\begin{definition}\label{def:semiGAS}
An equilibrium $e_\star$ of the error dynamics of a system-observer pair, on a manifold $\calM$, is \emph{almost globally stable} if its basin of attraction is the complement of an exception set $\chi \subset \calM$ of measure zero.

An equilibrium $e_\star$ of the error dynamics of a system-observer pair, on a manifold $\calM$, is said to be \emph{almost semi-globally stabilisable} if, for each compact set $K \subset \calM$ in the complement of an exception set $\chi \subset \calM$ of measure zero, there exists a choice of correction $\Delta(\hat{x},y)$ such that $e_\star$ is an asymptotically stable equilibrium of the error dynamics with basin of attraction containing $K$.
\end{definition}

\bibliographystyle{plainnat}
\bibliography{references}

\end{document}